\documentclass[letterpaper, 10 pt, conference]{ieeeconf}

\IEEEoverridecommandlockouts                              %

\usepackage{graphicx,amsmath,amsbsy,amssymb}
\usepackage{colortbl}	
\usepackage{tabularx}
\usepackage{theorem}	
\usepackage{cite}		
\usepackage{url}
\usepackage{xcolor}
\usepackage{comment}

\makeatletter
\g@addto@macro\normalsize{%
  \setlength{\abovedisplayskip}{6pt plus 2pt minus 3pt}%
  \setlength{\belowdisplayskip}{6pt plus 2pt minus 3pt}%
  \setlength{\abovedisplayshortskip}{3pt plus 1pt minus 2pt}%
  \setlength{\belowdisplayshortskip}{3pt plus 1pt minus 2pt}%
}
\makeatother

\theorembodyfont{\upshape}
\newtheorem{theorem}{Theorem}
\newtheorem{assumption}{Assumption}

\newtheorem{proposition}{Proposition}
\newtheorem{definition}{Definition}
\newtheorem{corollary}{Corollary}

\newcommand{\R}{\mathbb{R}}
\newcommand{\E}{\mathbb{E}}

\newcommand{\sqnormM}[2]{\left\|#1\right\|_{#2}^{2}}

\newcommand{\diag}{\mathrm{diag}}
\newcommand{\argmin}{\mathop{\mathrm{arg\,min}}}
\newcommand{\lmin}{\lambda_{\min}}

\title{Partial Observation Amplifies Model Mismatch in MAP Estimation \\
via Information-Curvature Margins}

\author{Junsei Ito and Yasuaki Wasa
\thanks{J.~Ito and Y.~Wasa are with the Department of Electrical Engineering and Bioscience, Waseda University, Tokyo 169-8555, JAPAN. 
{\tt\small \{distinction0625@moegi.,wasa@\}waseda.jp}}
\thanks{This work was partially supported by the Japan Science and Technology Agency (JST), ACT-X, Grant No.\ JPMJAX25C3.}
\thanks{\copyright~2026 IEEE. Personal use of this material is permitted. Permission from IEEE must be obtained for all other uses, in any current or future media, including reprinting/republishing this material for advertising or promotional purposes, creating new collective works, for resale or redistribution to servers or lists, or reuse of any copyrighted component of this work in other works.}
}

\begin{document}
\maketitle
\thispagestyle{empty}
\pagestyle{empty}

\begin{abstract}
This paper theoretically analyzes how system model mismatch displaces finite-horizon maximum a posteriori (MAP) initial-state estimates in controlled dynamical systems under partial observation.
From pathwise sensitivity analysis, the initial-state nominal--oracle displacement called MAP shift is decomposed into a model-side mismatch injection and an estimator-side curvature resistance to identify a sensor-dependent information-curvature margin as the amplification bottleneck.
The margin is governed by the weakest posterior-curvature direction, so that sensor configurations that maximize aggregate information can still be fragile to mismatches.
We connect the margin to nominal Gauss--Newton curvature and to the Bayesian Fisher information matrix, distinguishing instance-wise mismatch robustness from design-time inferability.
The margin admits a computable nominal proxy in nonlinear systems, becomes explicit in the linear time-invariant case, and is validated through two numerical examples.
\end{abstract}

%

%%%%%%%%%%%%%%%%%%%%%%%%%%%%%%%%%%%%%%%%%%%%%%%%%%%%%
%%%%%%%%%%%%%%%%%%%%%%%%%%%%%%%%%%%%%%%%%%%%%%%%%%%%%

\section{Introduction}
\label{sec:intro}

Data-driven control has attracted considerable attention as an effective approach to stabilizing and improving the performance of complex dynamical systems in engineering applications~\cite{dorfler2023CSM1,dorfler2023CSM2}. 
Grey-box modeling tools, including physics-informed neural networks~\cite{Raissi19} and sparse model discovery~\cite{Brunton16}, have expanded the range of nominal dynamical models available for state estimation and control design.
Residual system model mismatch, however, remains unavoidable due to parameter uncertainty and unmodeled dynamics.
In Sim2Real settings, this discrepancy is often called the reality gap~\cite{realitygap25}.
The impact of the model mismatch on estimation becomes particularly subtle under partial observation.
The primary reasons are not only the system norm of the mismatch but also the local posterior geometry estimated from the sensor configuration and the state-transition flow.

In this paper, we investigate the amplification mechanism in finite-horizon maximum a posteriori (MAP) initial-state estimation for deterministic dynamical systems with Gaussian prior and measurement noise under partial observation.
Rawlings et al. \cite{Rao03,Rawlings12} propose optimization-based state estimation based on full-information and moving-horizon formulations for nonlinear systems. 
Schiller et al. \cite{Schiller23} analyze robust stability of moving-horizon estimation under disturbances.
The papers \cite{Krener09,Shen18,Jujian15} analyze observability to quantify the difficulty of initial-state estimation through the observability Gramian and its empirical extensions for nonlinear systems.
In the papers \cite{VanTreesBell,Kunwoo23}, the Bayesian Fisher information matrix (BFIM) and the associated estimability Gramian extend the observability analysis by incorporating prior information into a design-time inferability metric.
Sayed \cite{Sayed01} addresses uncertainty-aware estimator design for robust filtering.
Bonnans and Shapiro \cite{Bonnans00} analyze how objective perturbations induce solution shifts in parameterized optimization.
However, none of these viewpoints provides an explicit instance-wise law from the system model mismatch to the displacement of a selected MAP solution under partial observation.
In particular, minimum-eigenvalue design criteria based on the observability Gramian and the BFIM quantify worst-direction inferability under a fixed model but do not bound the estimation bias induced by the system model mismatch.

To address this gap, this paper investigates the system model mismatch as a perturbation of the MAP first-order condition under partial observation.
By comparing a nominal estimator based on a known system function with an oracle estimator based on the unknown ideal function, a pathwise sensitivity argument identifies a sensor-dependent information-curvature margin as the amplification bottleneck.

The main contributions of this paper are as follows.
\begin{itemize}
\item We derive an instance-wise decomposition of the initial-state nominal--oracle displacement called MAP shift into a model-side mismatch injection and an estimator-side curvature resistance through a path-averaged posterior Hessian.
We identify the information-curvature margin as the amplification bottleneck (Theorems~\ref{thm:path_gain}--\ref{thm:mm_to_map}).
\item We show a structural bound from system model mismatch to MAP shift, with bridges from the pathwise margin to nominal Gauss--Newton curvature and the BFIM, distinguishing instance-wise robustness from design-time inferability (Theorem~\ref{thm:bfim}, Corollary~\ref{cor:proxy_bridge}).
\item The information-curvature margin admits a computable nominal proxy and a small-gain-type certificate (Corollary~\ref{cor:certificate}), and becomes explicit with a tight gain bound in the linear time-invariant case (Section~\ref{sec:lti}).
\item We numerically show that sensor configurations maximizing aggregate nominal information can remain fragile in the weakest curvature direction (Section~\ref{sec:exp}).
\end{itemize}

\paragraph*{Notations}
Throughout, $\|\cdot\|$ denotes the Euclidean norm for vectors and the spectral norm for matrices, $\|z\|_M:=\sqrt{z^\top M z}$ the weighted norm for $M\succ 0$, and $\lmin(\cdot)$ the minimum eigenvalue. The relation $A\succ 0$ ($A\succeq 0$) means that $A$ is symmetric positive definite (semidefinite). 
The function $f$ is of class $C^2$ if it is twice continuously differentiable.
For integers $i\le j$, let $z_{i:j}:=(z_i,z_{i+1},\ldots,z_j)$.

%%%%%%%%%%%%%%%%%%%%%%%%%%%%%%%%%%%%%%%%%%%%%%%%%%%%%
%%%%%%%%%%%%%%%%%%%%%%%%%%%%%%%%%%%%%%%%%%%%%%%%%%%%%

\section{Problem formulation}
\label{sec:prob}

We consider a discrete-time \emph{ideal mathematical model} that exactly reproduces a real physical plant with state $x_k^\star\in\mathcal{X}$ and control input $u_k\in\mathcal U\subset\R^m$:
\begin{align}
  x_{k+1}^\star &= f^\star(x_k^\star,u_k), \ k\in\mathbb{Z}, \label{eq:dyn_real}
\end{align}
where $\mathcal{U}$ is compact, and $\mathcal{X}\subset\R^n$ denotes a compact set containing all state trajectories and optimizer segments considered below.
This formulation is standard in model-based control and learning~\cite{Raissi19,Brunton16}.
We assume that $u_k$ is known, while the ideal system function $f^\star$ is unknown.

We also introduce a \emph{nominal mathematical model} with $u_k$, the measurement output $y_k\in\R^p$, and the estimated state $x_k\in\mathcal{X}$ at step $k\in\mathcal{K}:=\{0,1,\ldots,N{-}1\}$, $N\ge 1$:
\begin{align}
  x_{k+1} =& f^0(x_k,u_k), \ k\in\mathcal{K}, \quad 
  x_0 \sim \mathcal{N}(\mu_0,\Sigma_x), \label{eq:dyn} \\
  y_k =& h(x_k)+v_k, \quad 
  v_k \stackrel{\mathrm{i.i.d.}}{\sim} \mathcal{N}(0,R),\label{eq:obs}
\end{align}%
where $f^0$ is the known nominal system function available to the estimator. 
Here, $v_k$ is zero-mean Gaussian noise with covariance $R\,(\succ 0)$.
Assume $f^\star$ and $f^0$ are of class $C^2$ on $\R^n\times\mathcal{U}$, and $h$ is of class $C^2$ on $\R^n$.
For the input sequence $U:=u_{0:N-1}$ and relevant initial conditions, the state trajectories generated by both $f^\star$ and $f^0$ remain in $\mathcal{X}$ over the horizon.
Let $Y:=y_{0:N-1}$ denote the measurement sequence.

In this paper, we focus on partial observation, i.e., $p<n$, so that only part of the state is measured directly.
In practice, the input sequence is designed from the nominal model, whereas the physical plant evolves according to the unknown ideal function $f^\star$.
Our goal is to quantify how sensor selection via the limited measurement map $h(\cdot)$ affects the estimation impact of the system model mismatch under partial observation.
Due to the partial observation with measurement noise and the unknown ideal dynamics \eqref{eq:dyn_real}, the initial state $x_0$ is assigned the Gaussian prior in \eqref{eq:dyn}, with known mean $\mu_0$ and covariance matrix $\Sigma_x\,(\succ 0)$.
 
To evaluate this, we define a system model mismatch by
\begin{equation}
f_\Delta(x,u):=f^0(x,u)-f^\star(x,u), \label{eq:mismatch}
\end{equation}
which encompasses unmodeled dynamics and parametric uncertainty. 
The main challenge lies in analyzing how the unknown mismatch $f_\Delta$ propagates into the state estimation error under partial observation when estimating $x_0$ from the given input-output dataset $(Y,U)$.
To this end, we introduce the maximum a posteriori (MAP) estimate of $x_0$ under a given $f$, which is obtained by
\begin{align}
& \hat x_0(f) \in \argmin_{x_0\in\mathcal D}\ \Phi_f(x_0;Y,U), \label{eq:map}\\
& \Phi_f(x_0;Y,U) \notag \\
& \quad :=\frac{1}{2}\left[\sqnormM{x_0-\mu_0}{\Sigma_x^{-1}}
 +\sum_{k=0}^{N-1}
   \sqnormM{r_{f,k}(x_0;Y,U)}{R^{-1}}\right],
\label{eq:Phi}
\end{align}
where the innovation is defined as 
$r_{f,k}(x_0;Y,U):= y_k-h(\phi^{f,U}_k(x_0))$
with the transition flow described by
$\phi^{f,U}_{k+1}(x_0) := f(\phi^{f,U}_k\!(x_0), u_k)$, $k\in\mathcal{K}$, and $\phi^{f,U}_0\!(x_0) = x_0$.
The set $\mathcal D\subset\mathcal X$ is the search region of the MAP estimation in \eqref{eq:map}. Any convex and compact subset of $\mathcal X$ can serve as $\mathcal D$. The set $\mathcal D$ is not the support of the Gaussian prior.

In the sensitivity analysis, we compare a selected pair of nominal and oracle minimizers, denoted by $\hat x_0(f^0)$ and $\hat x_0(f^\star)$, that lie in $\mathrm{int}(\mathcal D)$.
When the selected minimizers are global over $\mathcal D$, they coincide with the restricted MAP estimates in~\eqref{eq:map}.
The nominal--oracle displacement
\begin{equation}
  \hat{x}_\Delta:=\hat{x}_0(f^0)-\hat{x}_0(f^\star)
  \label{eq:map_shift}
\end{equation}
is termed the MAP shift.
The MAP shift isolates the mismatch-induced component of the estimation error and excludes the oracle error caused by noise and prior.
Then, we introduce the posterior Hessian and the Bayesian Fisher information matrix (BFIM), also called the estimability Gramian~\cite{VanTreesBell,Kunwoo23}, in Definition~\ref{def:info}.
\begin{definition} \label{def:info}
For a function $f$, the posterior Hessian at $x_0$ is defined by $H_f(x_0;Y,U):=\nabla_{x_0}^2\Phi_f(x_0;Y,U)$.
Under the ideal function $f^\star$, the BFIM is defined by $G_N^o(U):= -\E[\partial^2 \log p_\star(x_0,Y\,|\,U)/\partial x_0^2]$,
where $p_\star(x_0,Y\,|\,U)$ is the joint density of $(x_0,Y)$ induced by the prior $x_0\sim\mathcal N(\mu_0,\Sigma_x)$, the true transition function $f^\star$, and the observation model~\eqref{eq:obs}.
\end{definition}
The matrix $H_f$ is the local curvature of the negative log-posterior.
When $H_f\succ 0$, it acts as a local information matrix that resists perturbations around $x_0$.
The BFIM is its design-time counterpart in expectation~\cite{VanTreesBell}.
A larger minimum eigenvalue of the BFIM implies a tighter Bayesian Cram\'er--Rao bound on the average estimation error~\cite{VanTreesBell,Kunwoo23}.

Under partial observation, information reaches $x_0$ only through the composition of the sensor map $h$ and the flow $\phi^{f,U}_k$, and the effect of the mismatch on estimation is no longer immediate.

%%%%%%%%%%%%%%%%%%%%%%%%%%%%%%%%%%%%%%%%%%%%%%%%%%%%%
%%%%%%%%%%%%%%%%%%%%%%%%%%%%%%%%%%%%%%%%%%%%%%%%%%%%%

\section{Main Results: Amplification Mechanism}
\label{sec:main}

The system model mismatch \eqref{eq:mismatch} does not determine the MAP shift by its size alone.
Under partial observation, the posterior curvature that resists the mismatch depends on the sensor configuration, namely on the measurement map $h$.
The mismatch can be harmless for some sensing configurations and fragile for others.
Since $\hat x_0(f^\star)$ satisfies the oracle first-order condition, the effect of using the nominal function $f^0$ enters through the mismatch injection
\begin{equation}
g_\Delta:=\nabla\Phi_{f^0}(\hat x_0(f^\star);Y,U).
\label{eq:Delta_g}
\end{equation}
Under this setting, we obtain the following results. 
Parts of the proofs are presented in the Appendix.

The objectives $\Phi_{f^0}(\cdot;Y,U)$ and $\Phi_{f^\star}(\cdot;Y,U)$ may admit multiple local minimizers in $\mathcal D$.
To analyze a selected pair of local minimizers, we introduce Assumption~\ref{ass:reg}.

\begin{assumption}\label{ass:reg}
For a given $(Y,U)$, $\Phi_{f^0}(\cdot;Y,U)$ is of class $C^2$ on an
open set containing the segment
$\{\hat x_0(f^\star)+s\hat{x}_\Delta:s\in[0,1]\}$, $\hat x_0(f^0)$ is a local minimizer of $\Phi_{f^0}(\cdot;Y,U)$, and $\hat x_0(f^\star)$ is a local minimizer of $\Phi_{f^\star}(\cdot;Y,U)$.
\end{assumption}
To analyze the sensitivity of the Hessian depending on the trajectory generated by the segment, we obtain Theorem~\ref{thm:path_gain} and Corollary~\ref{cor:weighted_gain}.
\begin{theorem}\label{thm:path_gain}
Under Assumption~\ref{ass:reg}, the equality
\begin{equation}
\bar H\,\hat{x}_\Delta=-g_\Delta
\label{eq:path_identity}
\end{equation}
with the path-averaged posterior Hessian
\begin{equation}
\bar H:=
\int_0^1
H_{f^0}\bigl(\hat x_0(f^\star)+s\hat{x}_\Delta;Y,U\bigr)\,ds
\label{eq:Hbar}
\end{equation}
holds. 
If $\lmin(\bar H)>0$, the inequality
\begin{equation}
\|\hat{x}_\Delta\|\le\mu\|g_\Delta\|
\label{eq:path_bound}
\end{equation}
holds, where $\mu:=1/\lmin(\bar H)$.
\end{theorem}
\begin{corollary}\label{cor:weighted_gain}
If $\lmin\!\bigl(W^{-1/2}\bar H W^{-1/2}\bigr)>0$ for any weight matrix $W\succ0$ under Assumption~\ref{ass:reg}, then $\|\hat{x}_\Delta\|_W\le\mu_W\|g_\Delta\|_{W^{-1}}$ holds, where $\mu_W:=1/\lmin\!\bigl(W^{-1/2}\bar H W^{-1/2}\bigr)$.
\end{corollary}
Theorem~\ref{thm:path_gain} separates the propagation mechanism into a
model-side perturbation $g_\Delta$ and an estimator-side curvature
margin $\lmin(\bar H)$.
The pathwise identity itself is a standard sensitivity
argument in parametric optimization~\cite{Bonnans00}.
The gain bound \eqref{eq:path_bound}, however, holds regardless of the direction of $g_\Delta$. The fragility of each sensor configuration to the system model mismatch is quantified by a single scalar, the information-curvature margin $\mu^{-1}=\lmin(\bar H)$. The rest of the paper makes this margin computable and exploits it for sensor design, using only standard arguments.
The weighted bound applies when state directions have different physical importance, and $\mu_W$ ranks worst-case weighted amplification, not the realized shift.

Calculating $\bar H$ directly is generally difficult. 
To analyze the computable case, we introduce the pathwise gain
\begin{equation}
\mu_{\mathrm{path}}:=
\Bigl(\min_{s\in[0,1]}
\lmin\!\bigl(
H_{f^0}\bigl(\hat x_0(f^\star)+s\hat{x}_\Delta;Y,U\bigr)
\bigr)\Bigr)^{-1}\!.
\label{eq:mu_path}
\end{equation}
Since $\bar H\succeq \mu_{\mathrm{path}}^{-1}I$,
Theorem~\ref{thm:path_gain} also yields
$\|\hat{x}_\Delta\|\le \mu_{\mathrm{path}}\|g_\Delta\|$
whenever $\mu_{\mathrm{path}}>0$.
The quantity $\mu_{\mathrm{path}}$ still depends on the unknown oracle
path.
Under Assumption~\ref{ass:gap_lip}, which holds for class $C^2$ functions on the compact $\mathcal{X}$~\cite{Khalil02}, Theorem~\ref{thm:mm_to_map} gives a structural mismatch-to-shift bound.

\begin{assumption}\label{ass:gap_lip}
The functions $f^0$, $f^\star$, and $h$ are uniformly differentiable and Lipschitz over the relevant compact sets.
There exist constants $\varepsilon, \varepsilon_J\ge 0$ such that the system model mismatch and its Jacobian satisfy
\begin{align*}
\sup_{x\in\mathcal X,u\in\mathcal U} \|f_\Delta(x,u)\|\le\varepsilon,
\ \sup_{x\in\mathcal X,u\in\mathcal U} \|\partial_xf_\Delta(x,u)\|\le\varepsilon_J.
\end{align*}
\end{assumption}
\begin{theorem}\label{thm:mm_to_map}
Under Assumptions~\ref{ass:reg} and~\ref{ass:gap_lip}, if
$\mu_{\mathrm{path}}>0$, there exist constants
$c_1(U),c_2(U),c_3(U)\ge0$, independent of
$\varepsilon$ and $\varepsilon_J$ and depending only on $N$, $U$, $R$, and the uniform derivative/Lipschitz bounds in Assumption~\ref{ass:gap_lip}, such that
\begin{align}
\|\hat{x}_\Delta\|
&\le\mu_{\mathrm{path}}
\Bigl(
c_1(U)\varepsilon + c_2(U)\varepsilon_J
\notag\\
&+ c_3(U)(\varepsilon{+}\varepsilon_J) \max_k\|R^{-1/2}r_{\star,k}(\hat x_0(f^\star))\| \Bigr),
\label{eq:mm_to_map}
\end{align}
where $r_{\star,k}(x_0):=r_{f^\star,k}(x_0;Y,U) = y_k-h(\phi^{f^\star,U}_k(x_0))$.
\end{theorem}

From Theorems~\ref{thm:path_gain} and \ref{thm:mm_to_map}, the system model mismatch propagates to the MAP shift through the chain $f_\Delta\to g_\Delta\to\hat{x}_\Delta$.
The remaining terms collect how the mismatch propagates through the dynamics and observation map.
Unlike the injection $\|g_\Delta\|$ in Theorem~\ref{thm:path_gain}, the constants $(\varepsilon,\varepsilon_J)$ in Assumption~\ref{ass:gap_lip} quantify the system model mismatch itself. The MAP shift scales linearly with $(\varepsilon,\varepsilon_J)$, and the sensor-dependent gain $\mu_{\mathrm{path}}$ sets the amplification factor.
The constants $c_i(U)$ can grow conservatively with the horizon, so the bound should be read as a structural propagation result rather than a numerically sharp certificate.
The residual factor in \eqref{eq:mm_to_map} is noise-level, not mismatch-level, since the oracle innovation is measurement noise up to estimation error.

To obtain computable proxies, we define $J_{f,k}(x_0):=\partial \phi^{f,U}_k(x_0)/\partial x_0$, $C_{f,k}(x_0):=(\partial h/\partial x)|_{x=\phi^{f,U}_k(x_0)}$, and the Gauss--Newton approximation of the posterior Hessian $G_{f,N}(x_0,U):=\Sigma_x^{-1}+\sum_{k=0}^{N-1}J_{f,k}^\top C_{f,k}^\top R^{-1}C_{f,k}J_{f,k}$.
Then, we obtain Theorem~\ref{thm:hess_decomp}.
\begin{theorem}\label{thm:hess_decomp}
For any $x_0\in\mathcal D$,
\begin{equation}
H_f(x_0;Y,U)
=
G_{f,N}(x_0,U)+S_f(x_0;Y,U),
\label{eq:hess_decomp}
\end{equation}
where $S_f(x_0;Y,U)$ collects residual-weighted second-derivative
terms.
If $\|S_f(x_0;Y,U)\|\le \delta_f(x_0;Y,U)$, then
\begin{equation*}
\lmin\bigl(H_f(x_0;Y,U)\bigr)
\ge
\lmin\bigl(G_{f,N}(x_0,U)\bigr)-\delta_f(x_0;Y,U).
\end{equation*}
\end{theorem}
The curvature originates in the trajectory information accumulated
through $C_{f,k}J_{f,k}$.
Under full-state observation, the sensor map does not discard any
state direction.
Under partial observation, however, weakly sensed directions can
survive through the flow while contributing little to the posterior
curvature.
The curvature margin $\lmin(G_{f,N})$ is therefore directly
controlled by the sensor map $C_{f,k}$.
The relation between Hessian $H_f$ and BFIM is shown in Theorem~\ref{thm:bfim}.

\begin{theorem}\label{thm:bfim}
For any given $U$ and any $x_0\in\mathcal D$,
\begin{align}
&\E\!\left[ H_{f^\star}(x_0;Y,U)\mid x_0,U \right]=
G_{f^\star,N}(x_0,U), \label{eq:thm4_01} \\
&\E\!\left[H_{f^\star}(x_0;Y,U)\right]
=
\E\!\left[G_{f^\star,N}(x_0,U)\right]
=
G_N^o(U), \label{eq:thm4_02}
\end{align}
where the expectation in~\eqref{eq:thm4_02} is taken over the joint law of $(x_0,Y)$ induced by $f^\star$ and $U$.
\end{theorem}
The BFIM $G_N^o(U)$ is the expected posterior curvature under the ideal dynamics, a design-time quantity for sensor and experiment design.
Unlike Theorems~\ref{thm:path_gain} and~\ref{thm:mm_to_map}, it is not a per-instance lower bound on the realized pathwise margin.

We next derive nominal proxies for the pathwise gain that avoid the oracle Hessian $\bar H$ and the oracle MAP point.
Define
\begin{align*}
\mu_{\rm nom}
&:= 1/\lmin\!\bigl(H_{f^0}(\hat x_0(f^0);Y,U)\bigr), \\
\mu_{\rm GN,nom}
&:= 1/\lmin\!\bigl(G_{f^0,N}(\hat x_0(f^0),U)\bigr).
\end{align*}
For a compact set $\mathcal B\subseteq\mathcal D$ containing the segment $[\hat x_0(f^\star),\,\hat x_0(f^0)]$, the modulus of continuity is defined by
\begin{equation*}
\omega_H(\eta):=
\!\!\!\! \sup_{x,x'\in\mathcal B, \|x-x'\|\le \eta} \!\!\!\!\!
\|H_{f^0}(x;Y,U)-H_{f^0}(x';Y,U)\|.
\end{equation*}
For $\mathcal B\subset\mathcal D$, we obtain Corollary~\ref{cor:proxy_bridge} on bounds of $\mu_{\mathrm{path}}^{-1}$.
\begin{corollary}\label{cor:proxy_bridge}
If $\mu_{\rm nom}^{-1}>0$ and $\mu_{\rm GN,nom}^{-1}>0$, then
\begin{equation}
\mu_{\rm nom}^{-1}-\omega_H(\|\hat{x}_\Delta\|)
\le
\mu_{\mathrm{path}}^{-1}
\le
\mu_{\rm nom}^{-1}.
\label{eq:proxy_bridge}
\end{equation}
Moreover, if
$\delta_{\rm nom}:=\|S_{f^0}(\hat x_0(f^0);Y,U)\|$, then
\begin{equation}
\mu_{\rm GN,nom}^{-1}
-\delta_{\rm nom}
-\omega_H(\|\hat{x}_\Delta\|) 
\le \mu_{\mathrm{path}}^{-1}.
\label{eq:proxy_bridge_GN}
\end{equation}
\end{corollary}
When the left-hand side of \eqref{eq:proxy_bridge_GN} is positive, \eqref{eq:proxy_bridge_GN} provides a lower bound on the margin $\mu_{\mathrm{path}}^{-1}$.
The lower bound in \eqref{eq:proxy_bridge_GN} contains the continuity term $\omega_H(\|\hat{x}_\Delta\|)$ and is explicit only when an external bound on $\|\hat{x}_\Delta\|$ is available.
Corollary~\ref{cor:certificate} removes the need for an external bound when $\omega_H$ is bounded by a linear function.
The linear bound follows from the definition of $\omega_H$ whenever the nominal Hessian $H_{f^0}(\cdot;Y,U)$ is Lipschitz on $\mathcal D$. Corollary~\ref{cor:proxy_bridge} allows any compact set $\mathcal B\subset\mathcal D$ containing the segment. Corollary~\ref{cor:certificate} takes $\mathcal B=\mathcal D$, so the constant $L_H$ does not depend on the unknown segment.

\begin{corollary}\label{cor:certificate}
Under Assumption~\ref{ass:reg}, we set $\mu_{\rm nom}^{-1}>0$, $\omega_H(\eta)\le L_H\eta$ for all $\eta\ge0$ with $\mathcal B=\mathcal D$ and a constant $L_H>0$.
Define $\varrho:=1/(\mu_{\rm nom}L_H)$ and suppose that the margin test $2L_H\mu_{\rm nom}^{2}\|g_\Delta\|\le 1$ holds.
Then $\|\hat{x}_\Delta\|\ge\varrho$ or
\begin{equation}
\|\hat{x}_\Delta\|\le 2\,\mu_{\rm nom}\|g_\Delta\|\le\varrho
\label{eq:certificate}
\end{equation}
holds.
If \eqref{eq:certificate} holds, then $\hat x_0(f^0)$ is the unique stationary point of $\Phi_{f^0}(\cdot;Y,U)$ on $\mathcal B_\varrho:=\{x\in\mathcal D:\|x-\hat x_0(f^0)\|<\varrho\}$.
\end{corollary}
Under the margin test, the shift $\|\hat{x}_\Delta\|$ never lies in the intermediate interval $(2\mu_{\rm nom}\|g_\Delta\|,\,\varrho)$.
The choice $\mathcal B=\mathcal D$ also makes the Lipschitz bound valid on the whole ball $\mathcal B_\varrho$, which the uniqueness statement in Corollary~\ref{cor:certificate} requires.
Every quantity in the margin test except $\|g_\Delta\|$ is computable from $(Y,U)$ and the nominal model, and Proposition~\ref{prop:gap_to_inj} bounds $\|g_\Delta\|$ under Assumption~\ref{ass:gap_lip}, so the margin test yields the near-branch bound under a prescribed mismatch level.
The constant $L_H$ depends only on the nominal model and $(Y,U)$, and any Lipschitz overestimate of the nominal Hessian over $\mathcal D$ is valid. Overestimating $L_H$ tightens the margin test and shrinks $\varrho$ but never invalidates \eqref{eq:certificate}.
In the near branch, the certificate also excludes basin ambiguity, since no other stationary point exists on $\mathcal B_\varrho$.
The far branch $\|\hat{x}_\Delta\|\ge\varrho$ remains possible, so the certificate is local.

%%%%%%%%%%%%%%%%%%%%%%%%%%%%%%%%%%%%%%%%%%%%%%%%%%%%%
%%%%%%%%%%%%%%%%%%%%%%%%%%%%%%%%%%%%%%%%%%%%%%%%%%%%%

\section{Special case: LTI systems}
\label{sec:lti}

Let us consider a linear time-invariant (LTI) system
\begin{align*}
    f^\star(x_k,u_k)&:=A^\star x_k+B^\star u_k, \\
    f^0(x_k,u_k)&:=A^0 x_k+B^0 u_k, \
    h(x_k):=Cx_k,
\end{align*}%
where $A^\star,A^0\in\R^{n\times n}$, $B^\star,B^0\in\R^{n\times m}$, and $C\in\R^{p\times n}$.
Then, the system model mismatch~\eqref{eq:mismatch} becomes
\begin{equation}
    f_\Delta(x,u)=(A^0{-}A^\star)x+(B^0{-}B^\star)u.   
\end{equation}
From Theorem~\ref{thm:path_gain}, we straightforwardly obtain Proposition~\ref{prop:lti_margin}.
\begin{proposition}\label{prop:lti_margin}
    For the LTI systems, the MAP objective
    $\Phi_{f^0}(x_0;Y,U)$ is quadratic in $x_0$ with constant Hessian
        $H^0=\Sigma_x^{-1}+\mathcal W_o^0(N)$,
        $\mathcal W_o^0(N):=\sum_{k=0}^{N-1}\bigl((A^0)^k\bigr)^\top C^\top R^{-1}C(A^0)^k$.
    Hence, the nominal--oracle MAP shift satisfies
    \begin{equation}
        \hat x_\Delta = -(H^0)^{-1}g_\Delta,\qquad
        \|\hat x_\Delta\|\le \mu_{\mathrm{LTI}}\|g_\Delta\|,
        \label{eq:lti_shift_bound}
    \end{equation}
    where $\mu_{\mathrm{LTI}}:=1/\lmin(H^0)$.
\end{proposition}
From Proposition~\ref{prop:lti_margin}, the system model mismatch affects the MAP shift through the injection $g_\Delta$, whereas the sensor configuration affects the amplification only through the nominal curvature matrix $H^0$.
The oracle objective built from $(A^\star,B^\star)$ also has a constant Hessian, but the amplification in \eqref{eq:lti_shift_bound} is governed by the nominal curvature $H^0$, because the estimator is constructed from the nominal model.

The worst case is attained when $g_\Delta$ lies in the eigenspace of $\lmin(H^0)$, so \eqref{eq:lti_shift_bound} is tight and maximizing the margin $\mu_{\mathrm{LTI}}^{-1}$ is the exact minimax sensor design for a fixed injection budget $\|g_\Delta\|$.
To reduce the largest estimation variance under a correctly specified model, classical E-optimal design maximizes the smallest eigenvalue of an information matrix~\cite{Pronzato08}.
The margin justifies the same criterion by the tight worst-case bias bound under mismatch.

A similar argument to Corollary~\ref{cor:proxy_bridge} yields Proposition~\ref{prop:sensor_addition}.
\begin{proposition}\label{prop:sensor_addition}
    Consider the same setting as Proposition~\ref{prop:lti_margin} with diagonal $R$.
    Let $C_{\mathrm{part}}$ be obtained by removing rows of $C_{\mathrm{full}}$ and the corresponding entries of $R$.
    Then
        $H_{\mathrm{full}}=H_{\mathrm{part}}+\mathcal W_{o,\mathrm{lost}}(N)$, %\\
        $H_{\mathrm{full}}\succeq H_{\mathrm{part}}\succ 0$, %\
        $\mathcal W_{o,\mathrm{lost}}(N)\succeq 0$,
    where
    $\mathcal W_{o,\mathrm{lost}}(N):=\mathcal W_{o,\mathrm{full}}(N)-\mathcal W_{o,\mathrm{part}}(N)$.
    Therefore, the information-curvature margins
    satisfy
        $\mu_{\mathrm{full}}^{-1}\ge \mu_{\mathrm{part}}^{-1}$,
    where
    $\mu_{\bullet}:=1/\lmin(H_{\bullet})$.
    If additionally
    $\lmin(\mathcal W_{o,\mathrm{lost}}(N))>0$, then
    $\mu_{\mathrm{full}}^{-1}\ge \mu_{\mathrm{part}}^{-1}+\lmin\!\bigl(\mathcal W_{o,\mathrm{lost}}(N)\bigr)$.
\end{proposition}
The mismatch-amplification mechanism of LTI systems is explicit. 
The mismatch enters through $g_\Delta$, while partial observation and sensor placement determine the curvature-side resistance through $\mu_{\mathrm{LTI}}^{-1}$.
Adding sensors can only increase the margin and decrease the gain.

%%%%%%%%%%%%%%%%%%%%%%%%%%%%%%%%%%%%%%%%%%%%%%%%%%%%%
%%%%%%%%%%%%%%%%%%%%%%%%%%%%%%%%%%%%%%%%%%%%%%%%%%%%%
\section{Numerical validation}
\label{sec:exp}

\subsection{Validating pathwise gain bound and nominal proxy}

We first validate Theorem~\ref{thm:path_gain} and Corollary~\ref{cor:proxy_bridge} through a nonlinear damped pendulum.
The dynamics with state $x=[\theta,\dot\theta]^\top$ are $\ddot{\theta}=-(g/L)\sin\theta-c\,\dot\theta+u$. Only the angle is observed through $y=\theta+v$ with $v\sim\mathcal{N}(0,\sigma_v^2)$.
The system is discretized using the forward Euler method with a time step of $0.02$\,s over $N=80$ steps.
The ideal parameters are $(g/L,c)=(9.81,0.15)$ and the nominal model uses $(g_0/L_0,c_0)=(0.9\,g/L,\,1.3\,c)$.
For each of 120 Monte Carlo trials, an ideal initial state
$x_0\sim\mathcal{N}(0,\diag(0.6^2,1.2^2))$, control inputs
$u_k\stackrel{\mathrm{i.i.d.}}{\sim}\mathrm{Uniform}[-1,1]$, and noise $\sigma_v=0.05$ are generated.
Both MAP estimates $\hat x_0(f^\star)$ and $\hat x_0(f^0)$ are computed
by BFGS, with the nominal solve warm-started at the oracle optimizer
to track the local branch targeted by the theory.
The warm start does not affect the results. Nominal solves without the warm start converge to the same minimizers in all 120 trials.

The pathwise amplification gain $\mu_{\mathrm{path}}$ is approximated using 11 equally spaced points on the segment from
$\hat x_0(f^\star)$ to $\hat x_0(f^0)$.
All points in Fig.~\ref{fig:pendulum_bounds}(a) lie below the diagonal,
which empirically supports the gain bound
$\|\hat{x}_\Delta\|\le \mu_{\mathrm{path}}\|g_\Delta\|$.
The nominal Gauss--Newton gain proxy $\mu_{\rm GN,nom}$ also provides a
useful surrogate when the realized MAP shift is small, as shown in Fig.~\ref{fig:pendulum_bounds}(b).

\begin{figure}[t]
  \centering
  \includegraphics[clip,bb=14 9 261 107,width=\linewidth]{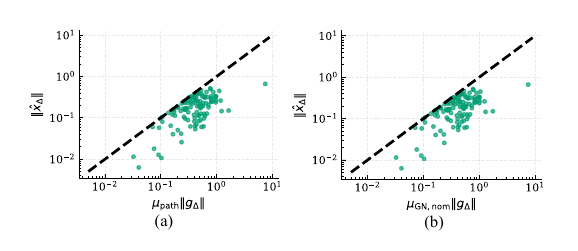}
  \caption{Relation with MAP shift $\|\hat{x}_\Delta\|$ in nonlinear pendulum. (a) Pathwise gain bound $\mu_{\mathrm{path}}\|g_\Delta\|$ and (b) Nominal GN gain proxy $\mu_{\rm GN,nom}\|g_\Delta\|$.}
  \label{fig:pendulum_bounds}
\end{figure}

\subsection{Curvature margin as a robustness predictor}

We next test whether the information-curvature margin in Section~\ref{sec:lti} predicts realized robustness to the system model mismatch. The experiments evaluate and rank given sensor placements. The gain bound covers every placement, and the search for a best placement is a separate design problem. The testbed is a mass--spring--damper (MSD) chain of $n_M=10$ masses with state $x=[q^\top,\dot{q}^\top]^\top\in\R^{2n_M}$ and dynamics
\begin{align}
  \dot{x} &= \begin{bmatrix} 0 & I_{n_M} \\
  -M^{-1}K & -M^{-1}D \end{bmatrix} x
  + \begin{bmatrix} 0 \\ M^{-1}B_u \end{bmatrix} u,
  \label{eq:chain_ct}\\
  y &= Cx + v,
  \quad v\sim\mathcal{N}(0,I\sigma_v^2), \notag
\end{align}
where $M=\diag(m_1,\dots,m_{n_M})$ is the mass matrix, $K,D\in\R^{n_M\times n_M}$ are the symmetric tridiagonal stiffness and damping matrices of the fixed-free topology, and $C$ selects $n_s$ of the $n_M$ position channels.
The system is discretized with a time step of $0.01$\,s.
The nominal parameters are $(m_0,c_0,k_0)=(1.0,0.4,20.0)$.
The ideal parameters are obtained by linear interpolation $(1-\beta)(m_0,c_0,k_0)+\beta(m_1,c_1,k_1)$ with $(m_1,c_1,k_1)=(1.4,0.7,40.0)$, so that $\beta=0$ gives no mismatch and $\beta=1$ gives the largest mismatch.
The prior is $x_0\sim\mathcal{N}(0,10I)$, the input is persistently exciting over $N=200$ steps, and $\sigma_v=0.05$.

\begin{figure}[t]
  \centering
  \includegraphics[clip,bb=14 8 264 106,width=\linewidth]{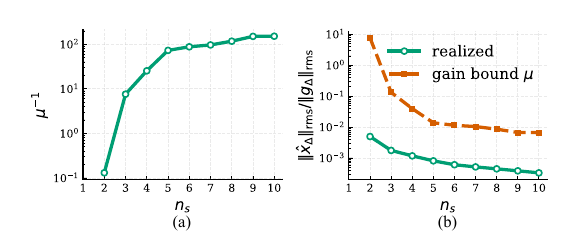}
  \caption{Relationship with the number of sensors $n_s$ in MSD system. (a) Information-curvature margin
$\mu_{\mathrm{LTI}}^{-1}$.
  (b) Realized amplification factor
  $\|\hat{x}_\Delta\|_{\mathrm{rms}}/\|g_\Delta\|_{\mathrm{rms}}$, together with the gain bound $\mu_{\mathrm{LTI}}$ (dashed).}
  \label{fig:sensor_count}
\end{figure}

\begin{figure}[t]
  \centering
  \includegraphics[clip,bb=11 15 331 126,width=\linewidth]{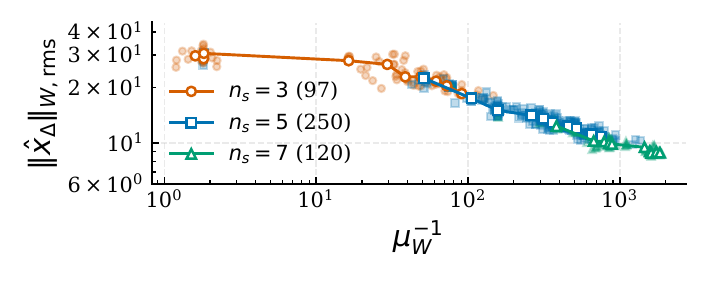}
  \caption{Weighted RMS shift $\|\hat{x}_\Delta\|_{W,\mathrm{rms}}$ versus weighted margin $\mu_W^{-1}$ for all observable placements with $n_s\in\{3,5,7\}$ at $\beta=1$.}

  \label{fig:sensor_ranking}
\end{figure}

The results in Fig.~\ref{fig:sensor_count} confirm Proposition~\ref{prop:sensor_addition}.
All placements share the same nominal/ideal-model pair, mismatch level $\beta=1$, and common Monte Carlo samples, ensuring that differences reflect only sensor choice.
The selected $n_s$ position sensors represent the margin-maximizing observable placement.
Shift and gradient norms are reported as root-mean-square (RMS) values over 80 trajectories per placement.
From Fig.~\ref{fig:sensor_count}(a), the margin $\mu_{\mathrm{LTI}}^{-1}$ increases monotonically with $n_s$, while the realized amplification factor $\|\hat{x}_\Delta\|_{\mathrm{rms}}/\|g_\Delta\|_{\mathrm{rms}}$ decreases with $n_s$ and remains below the gain bound $\mu_{\mathrm{LTI}}$ as shown in Fig.~\ref{fig:sensor_count}(b).

The second experiment tests whether the weighted margin $\mu_W^{-1}$
correlates with realized robustness across sensor placements within a
fixed budget.
For $n_s\in\{3,5,7\}$, all observable placements are enumerated and
evaluated on a common 50 Monte Carlo trajectories.
The weighting $W=\diag(I_{n_M},\omega_0^{-2}I_{n_M})$ with $\omega_0=\sqrt{k_0/m_0}$ nondimensionalizes the position--velocity state as in Corollary~\ref{cor:weighted_gain}.
From Fig.~\ref{fig:sensor_ranking}, larger $\mu_W^{-1}$ is strongly associated with smaller weighted shift, and the within-budget Spearman correlations are $\rho=-0.84,-0.94,-0.86$ for $n_s=3,5,7$.
Within these rankings, aggregate nominal information and mismatch robustness can disagree.
A trace-favoring placement, as in empirical-observability criteria~\cite{Qi16}, accumulates large nominal information but can leave a weak curvature direction.
A margin-favoring placement instead strengthens the weakest weighted direction, and its weighted shift is 1.9--2.8$\times$ smaller (Fig.~\ref{fig:placement_comparison}).
The margin quantifies the curvature-side resistance in the gain bound, not the mismatch itself.

\begin{figure}[t]
  \centering
  \includegraphics[clip,bb=5 6 396 249,width=\linewidth]{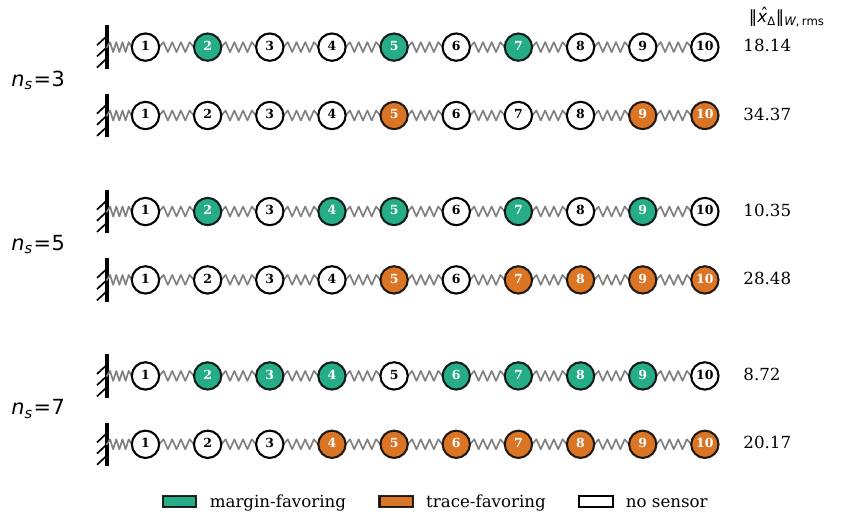}
  \caption{Trace-favoring and margin-favoring sensor placements on the
  10-mass chain for $n_s\in\{3,5,7\}$.
  Colored circles indicate sensed masses.}
  \label{fig:placement_comparison}
\end{figure}

%%%%%%%%%%%%%%%%%%%%%%%%%%%%%%%%%%%%%%%%%%%%%%%%%%%%%
%%%%%%%%%%%%%%%%%%%%%%%%%%%%%%%%%%%%%%%%%%%%%%%%%%%%%

\section{Conclusion}
\label{sec:conclusion}

This paper has theoretically shown how partial observation can amplify the impact of system model mismatch on MAP initial-state estimation via a sensor-dependent information-curvature margin.
The estimator-side sensitivity is determined by the weakest posterior curvature direction, not by aggregate information alone.
The information-curvature margin connects to nominal Gauss--Newton curvature and to the Bayesian Fisher information matrix, admits a computable nominal proxy in nonlinear systems, and becomes explicit in the LTI case.
The numerical experiments confirmed the pathwise gain bound, the nominal proxy, and the predicted sensor effects.
Future work will address process-noise models, sensor design under explicit model uncertainty, and margin-aware data-driven control under partial observation.

We acknowledge using OpenAI's ChatGPT and Anthropic's Claude to improve the manuscript readability.

\bibliographystyle{IEEEtran}
\bibliography{IEEEabrv,refs}

\appendix

%%%%%%%%%%%%%%%%%%%%%%%%%%%%%%%%%%%%%%%%%%%%%%%%%%%%%
%%%%%%%%%%%%%%%%%%%%%%%%%%%%%%%%%%%%%%%%%%%%%%%%%%%%%

\section{Proofs}
\label{sec:proof}

\noindent\textit{Proof of Theorem~\ref{thm:path_gain}.}
Define $x_s:=\hat x_0(f^\star)+s\hat{x}_\Delta$ for $s\in[0,1]$
and $F(s):=\nabla\Phi_{f^0}(x_s;Y,U)$.
By the chain rule, $F'(s)=H_{f^0}(x_s;Y,U)\,\hat{x}_\Delta$.
Integrating over $s\in[0,1]$ and using~\eqref{eq:Hbar} give
$F(1)-F(0)=\bar H\,\hat{x}_\Delta$.
By nominal optimality $F(1)=0$ and \eqref{eq:Delta_g}, $F(0)=g_\Delta$ and \eqref{eq:path_identity} hold.
If $\lmin(\bar H)>0$, $\hat{x}_\Delta=-\bar H^{-1}g_\Delta$ and \eqref{eq:path_bound} are obtained.
\hfill$\blacksquare$

\medskip
\noindent\textit{Proof of Corollary~\ref{cor:weighted_gain}.}
Let $\tilde H:=W^{-1/2}\bar H W^{-1/2}$ and $\tilde z:=W^{-1/2}z$.
Then $\|\bar H^{-1}z\|_W=\|\tilde H^{-1}\tilde z\|$ and
$\|z\|_{W^{-1}}=\|\tilde z\|$, so
$\sup_{z\neq0}\|\bar H^{-1}z\|_W/\|z\|_{W^{-1}}
=\|\tilde H^{-1}\|=1/\lmin(\tilde H)=\mu_W$
by the spectral-norm identity for positive definite matrices.
Applying this to $\hat{x}_\Delta=-\bar H^{-1}g_\Delta$
gives the weighted bound in Corollary~\ref{cor:weighted_gain}.
\hfill$\blacksquare$

\begin{proposition}\label{prop:gap_to_inj}
Under Assumption~\ref{ass:gap_lip}, for any $x_0\in\mathcal D$ there
exist constants $c_1(U),c_2(U),c_3(U)\ge0$, independent of
$\varepsilon$ and $\varepsilon_J$ and depending only on $N$, $U$, $R$, and the uniform derivative/Lipschitz bounds of $f^0$, $f^\star$, and $h$ such that
\begin{align}
&\|\nabla\Phi_{f^0}(x_0;Y,U)-\nabla\Phi_{f^\star}(x_0;Y,U)\| 
\le c_1(U)\varepsilon + c_2(U)\varepsilon_J
\notag\\
&\quad
+ c_3(U)(\varepsilon+\varepsilon_J)
\max_k\|R^{-1/2}r_{\star,k}(x_0)\|.
\label{eq:gap_to_inj}
\end{align}
\end{proposition}
\noindent\textit{Proof of Proposition~\ref{prop:gap_to_inj}.}
Applying the discrete Gr\"onwall inequality to the trajectory difference
$\delta x_k:=\phi_k^{f^\star,U}(x_0)-\phi_k^{f^0,U}(x_0)$,
which satisfies $\|\delta x_{k+1}\|\le L_f\|\delta x_k\|+\varepsilon$ with $\delta x_0=0$ and $L_f$  the Lipschitz constant of $f^\star$ in $x$, yields $\|\delta x_k\|=O(\varepsilon)$.
A second application to the flow-Jacobian difference $\delta J_k:=J_{f^\star,k}(x_0)-J_{f^0,k}(x_0)$, whose recursion additionally involves $\varepsilon_J$, yields $\|\delta J_k\|=O(\varepsilon)+O(\varepsilon_J)$.
Substituting these into the gradient difference and decomposing it into residual and sensitivity terms yields~\eqref{eq:gap_to_inj}.
\hfill$\blacksquare$

\medskip
\noindent\textit{Proof of Theorem~\ref{thm:mm_to_map}.}
Applying Proposition~\ref{prop:gap_to_inj} at
$x_0=\hat x_0(f^\star)$ bounds $\|g_\Delta\|$
by the right-hand side of~\eqref{eq:gap_to_inj}.
Since $\lmin(\bar H)\ge \mu_{\mathrm{path}}^{-1}$
by~\eqref{eq:mu_path},
Theorem~\ref{thm:path_gain} yields~\eqref{eq:mm_to_map}.
\hfill$\blacksquare$

\medskip
\noindent\textit{Proof of Theorem~\ref{thm:hess_decomp}.}
From~\eqref{eq:Phi},
$\Phi_f(x_0;Y,U) =
\frac{1}{2}\|x_0-\mu_0\|_{\Sigma_x^{-1}}^2
+
\frac{1}{2}\sum_{k=0}^{N-1} \|r_{f,k}(x_0)\|_{R^{-1}}^2$ holds.
For each $k$,
$\frac{\partial r_{f,k}}{\partial x_0}
= -C_{f,k}(x_0)\,J_{f,k}(x_0)$.
Differentiating
$\frac{1}{2}\|r_{f,k}(x_0)\|_{R^{-1}}^2$
twice gives the Gauss--Newton term
$J_{f,k}(x_0)^\top
C_{f,k}(x_0)^\top
R^{-1}
C_{f,k}(x_0)
J_{f,k}(x_0)$,
plus second-derivative terms weighted by the residual
$r_{f,k}(x_0)$.
Collecting these 
terms over
$k\in\mathcal{K}$ defines $S_f(x_0;Y,U)$ and yields \eqref{eq:hess_decomp}.
The lower bound follows from Weyl's inequality.
By definition, $G_{f,N}(x_0,U)$ and $S_f(x_0;Y,U)$ are symmetric and $\lmin(G_{f,N}(x_0,U)+S_f(x_0;Y,U))\ge \lmin(G_{f,N}(x_0,U))-\|S_f(x_0;Y,U)\|$ holds.
\hfill$\blacksquare$

\medskip
\noindent\textit{Proof of Theorem~\ref{thm:bfim}.}
Apply Theorem~\ref{thm:hess_decomp} with $f=f^\star$.
Conditioned on $(x_0,U)$, the trajectory $\phi_k^{f^\star,U}(x_0)$ is deterministic and~\eqref{eq:obs}
with $\E[v_k\mid x_0,U]=0$.
Hence
$\E[r_{f^\star,k}(x_0;Y,U)\mid x_0,U]=0$
for each $k\in\mathcal{K}$, so the conditional expectation of the
residual-weighted term $S_{f^\star}(x_0;Y,U)$ vanishes,
which proves~\eqref{eq:thm4_01}.
Taking the expectation of both sides yields
$\E\!\left[ H_{f^\star}(x_0;Y,U) \right]=\E\!\left[ G_{f^\star,N}(x_0,U) \right]$.
Since $\Phi_{f^\star}(x_0;Y,U)$ equals
$-\log p_\star(x_0,Y\mid U)$ up to an additive term independent of $x_0$,
their Hessians 
coincide.
Hence, by Definition~\ref{def:info}, $
G_N^o(U)=\E\!\left[ H_{f^\star}(x_0;Y,U) \right]$,
which proves~\eqref{eq:thm4_02}.
\hfill$\blacksquare$

\medskip
\noindent\textit{Proof of Corollary~\ref{cor:proxy_bridge}.}
The upper bound uses \eqref{eq:mu_path} for $\hat x_0(f^0)$ with $s=1$.
For the lower bound, Weyl's inequality gives
$\lmin(H_{f^0}(x^s;Y,U)) \ge \lmin(H_{f^0}(\hat x_0(f^0);Y,U))
- \|H_{f^0}(x^s;Y,U)-H_{f^0}(\hat x_0(f^0);Y,U)\|$,
and $\|x^s-\hat x_0(f^0)\|\le \|\hat{x}_\Delta\|$ for
$x^s:=\hat x_0(f^\star)+s\hat{x}_\Delta$.
The Gauss--Newton bound follows from
Theorem~\ref{thm:hess_decomp} at $x_0=\hat x_0(f^0)$.
\hfill$\blacksquare$

\medskip
\noindent\textit{Proof of Corollary~\ref{cor:certificate}.}
Assume $t:=\|\hat{x}_\Delta\|<\varrho$; otherwise the claim holds.
Averaging Weyl's inequality and $\omega_H(\eta)\le L_H\eta$ over the segment gives $\lmin(\bar H)\ge\mu_{\rm nom}^{-1}-L_Ht/2>0$, so Theorem~\ref{thm:path_gain} yields $(L_H/2)t^2-\mu_{\rm nom}^{-1}t+\|g_\Delta\|\ge0$.
Under the margin test the roots $t_\pm$ are real with $t_+\ge\varrho>t$, and rationalizing $t_-$ gives $t\le t_-\le2\mu_{\rm nom}\|g_\Delta\|\le\varrho$.
The same curvature bound gives $H_{f^0}(x;Y,U)\succ0$ on $\mathcal B_\varrho$, hence at most one stationary point.
\hfill$\blacksquare$
\end{document}